\documentclass[12pt]{article}

\usepackage{amsmath}
\usepackage{amssymb}
\usepackage{amsfonts}
\usepackage{latexsym}
\usepackage{color}
\usepackage{xcolor}
\usepackage{graphicx}

\catcode `\@=11 \@addtoreset{equation}{section}

\catcode `\@=12

\newcommand{\be}{\begin{equation}}
\newcommand{\en}{\end{equation}}
\newcommand{\bea}{\begin{eqnarray}}
\newcommand{\ena}{\end{eqnarray}}
\newcommand{\beano}{\begin{eqnarray*}}
\newcommand{\enano}{\end{eqnarray*}}
\newcommand{\bee}{\begin{enumerate}}
\newcommand{\ene}{\end{enumerate}}

\newcommand{\R}{\mathbb{R}}
\newcommand{\mc}{\mathcal}

\newcommand{\D}{{\mc D}}

\newcommand{\Sc}{{\cal S}}

\newcommand{\F}{{\cal F}}

\newcommand{\Lc}{{\cal L}}

\newcommand{\1}{1 \!\! 1}

\newcommand{\ii}{\mathrm{i}}
\newcommand{\ud}{\mathrm{d}}

\newcommand{\Hil}{\mc H}

\newtheorem{thm}{Theorem}

\newtheorem{prop}[thm]{Proposition}
\newtheorem{defn}[thm]{Definition}

\newenvironment{proof}{\noindent {\bf Proof --}}{\hfill$\square$ \vspace{3mm}\endtrivlist}

\catcode `\@=11 \@addtoreset{equation}{section}
\catcode `\@=12

\begin{document}

\thispagestyle{empty}

\vspace*{1cm}

\begin{center}
{\Large \bf Generalized Heisenberg algebra and (non linear) pseudo-bosons}   \vspace{1.5cm}\\

{\large F. Bagarello}\\
 DEIM -Dipartimento di Energia, ingegneria dell' Informazione e modelli Matematici,
\\ Scuola Politecnica, Universit\`a di Palermo, I-90128  Palermo, Italy\\
and\\ INFN, Sezione di Napoli.\\
e-mail: fabio.bagarello@unipa.it

\vspace*{.5cm}

{\large E. M. F. Curado}\\
Centro Brasileiro de Pesquisas Fisicas
and  Instituto Nacional de Ci\^encia e Tecnologia - Sistemas Complexos\\
Rua Xavier Sigaud 150, 22290-180 - Rio de Janeiro, RJ, Brazil\\
e-mail: evaldo@cbpf.br

 \vspace*{.5cm}

{\large J. P. Gazeau}\\
APC, UMR 7164, Univ Paris  Diderot, Sorbonne Paris Cit\'e,
75205 Paris, France\\
and\\
Centro Brasileiro de Pesquisas Fisicas\\
Rua Xavier Sigaud 150, 22290-180 - Rio de Janeiro, RJ, Brazil\\
e-mail: gazeau@apc.in2p3.fr

\end{center}

\vspace{3mm}

\begin{abstract}
\noindent We propose a deformed version of the generalized Heisenberg algebra by using techniques borrowed from the theory of pseudo-bosons. In particular, this analysis is relevant when non self-adjoint Hamiltonians are needed to describe a given physical system. We also discuss relations with nonlinear pseudo-bosons. Several examples are discussed.

\end{abstract}

%\vspace{2cm}

%{\bf PACS Numbers}:  .......

\vfill

%\pagenumbering{roman}

\newpage

\section{Introduction and preliminary results}\label{sectI}

The problem of finding exactly solvable quantum systems has always raised a lot of interest in the community of physicists and mathematicians (see for instance \cite{carinena14} and references therein). The reason is obvious: from a physical side, a quantum system is meant to describe some (relevant) phenomenon, mainly at microscopic scale. Solving the Schr\"odinger equation with a potential not considered so far can represent a serious mathematical challenge.
With this in mind, along the years several general ideas have been proposed to produce new solvable models out of old ones. For instance, if one knows the eigenvalues and  eigenvectors of a given Hamiltonian $H$, with $H=H^\dagger$, one can construct new related Hamiltonians related to $H$ by some similarity operator, or by intertwining operators, or, if $H$ can be Darboux factorized, by considering its super-symmetric partner. In each of these ways these new Hamiltonians have spectra and eigenvectors which can be deduced from the ones of $H$. It should also be stressed that, in recent years, and in particular since the publication of the famous paper by Bender and Boettcher in 1998, \cite{ben1}, similar strategies have been extended also to the case in which the original Hamiltonian $H$ is not necessarily self-adjoint, $H\neq H^\dagger$.

In recent years an alternative strategy to construct a solvable model working with self-adjoint Hamiltonians has been proposed by one of us (E.C.) and his coworkers in a series of papers on what has been called {\em generalized Heisenberg algebras} (GHA), \cite{curado,curado1,curado2}. This strategy is mainly based on the existence of suitable intertwining and commutation relations, and on the existence of a certain function related to them. On a different side, deformations of the canonical (anti-)commutation relations have proved to be quite useful for deducing eigenvalues and eigenvectors of certain Hamiltonians appearing in the literature  devoted to non-Hermitian quantum mechanics. We refer to \cite{baginbagbook} for a recent review. Here we somehow merge the two approaches showing how a GHA can be deformed to include in the framework Hamiltonians which are manifestly non self-adjoint, and how eigenvalues and biorthogonal eigenvectors can be constructed explicitly.

The paper is organized as follows: in the rest of this section we briefly discuss some essential aspects of the GHA. In Section \ref{sectII} we propose our deformed version of this algebra. We present in Section \ref{sectIII} several examples. Those   devoted to the P\"oschl-Teller potential and its limit case which is the infinite square-well potential are examined in detail. In Section \ref{sectIV} we discuss the relation of the deformed GHA with the so-called non linear $\D$-pseudo bosons, \cite{fbnlpbs}, while Section \ref{sectconcl} contains our conclusions. To keep the paper self-contained, the Appendix contains some useful formulas on the Gegenbauer polynomials and definitions on biorthogonal sets used in the paper.

\subsection{An introduction to GHA}\label{sectAITGHA}

Let  $c$ and $H=H^\dagger$ be two operators satisfying the following conditions:
\be
cH=f(H)c,\qquad\mbox{and}\qquad [c,c^\dagger]=f(H)-H,
\label{11}\en
where $f(x)$ is a given function of $x$, known as {\em the characteristic function} of the GHA. It has to be stressed that both these equalities should be properly defined, at least if $c$ or $H$, or both, are unbounded operators. This is the case, for instance, in the simple case when $H=c^\dagger c$, with $[c,c^\dagger]=\1$. When this is so, formulas (\ref{11}) should be understood in the sense of unbounded operators. For instance, following \cite{baginbagbook}, we can assume that a set $\D$ exists, dense in the Hilbert space $\Hil$, stable under the action of $c$, $c^\dagger$, $H$ and $f(H)$ such that
$$
cH\varphi=f(H)c\,\varphi, \qquad \mbox{and} \qquad (c\,c^\dagger-c^\dagger\,c)\varphi=(f(H)-H)\varphi,
 $$
for all $\varphi\in\D$. Of course, $f(H)$ can be naturally defined using functional calculus, since $H=H^\dagger$, \cite{rs}. From now on, when no confusion can arise, we will use the simpler notation in (\ref{11}).

From (\ref{11}) we easily deduce that $[c,H]=(f(H)-H)c$, and $[c^\dagger,H]=-c^\dagger(f(H)-H)$. Let us suppose $H$ positive,
and having  $\hat e_0 \in \mathcal{H}$ as a normalized ground-state of $H$, $H \hat e_0=\epsilon_0 \, \hat e_0$.
For concreteness we assume $\epsilon_0>0$. If  $\hat e_0$ belongs to $\D$, we can introduce the vectors $\hat e_n=(c^\dagger)^n\hat e_0$, $n\geq0$, and they all belong to $\D$ as well, and they are all eigenstates of $H$ with eigenvalues $\epsilon_n$, defined recursively as $\epsilon_n=f(\epsilon_{n-1})$, $n\geq1$. In other words, we have
\be
H\hat e_n=\epsilon_n\hat e_n,
\label{12}\en
$n\geq0$, with $\epsilon_n$ as above. Of course, self-adjointness of $H$ implies that, if the multiplicity of each $\epsilon_n$ is one, different $\hat e_n$ are mutually orthogonal: $\left<\hat e_n,\hat e_k\right>=0$ if $n\neq k$. However they are, in general, not normalized. To fix the normalization it is useful to check first if $c$ is a lowering operator, i.e. if $c\hat e_n$ is proportional to $\hat e_{n-1}$. For that it is necessary to check first that $c\hat e_0=0$.  This may follow from an explicit computation, but it may also be deduced using the following general result:

\begin{prop}\label{prop0}
Let $\F_{\hat e}=\{\hat e_n,\,n\geq0\}$ be complete in $\Hil$. Then $c\, \hat e_0=0$.
\end{prop}
\begin{proof}
Using the orthogonality of the vectors $\hat e_n$ and their definition we have
$$
\left<\hat e_n,(c\hat e_0)\right>=\left<c^\dagger\hat e_n,\hat e_0\right>=\left<\hat e_{n+1},\hat e_0\right>=0
$$
for all $n\geq0$. Then the result follows from our assumption on $\F_{\hat e}$.

\end{proof}

Using now induction on $n$ we can prove that, if $c\, \hat e_0=0$, then, for $n\geq1$,
\be
c\, \hat e_n=\sqrt{\epsilon_n-\epsilon_0} \,\hat e_{n-1}, \qquad \mbox{ and }\qquad  c^\dagger \,c\, \hat e_n=(\epsilon_n-\epsilon_0)\hat e_{n}.
\label{13}\en
The proof is easy and is based on (\ref{11}), and will not be given here. Rather than this we comment that, in order for the framework to make sense, the function $f(x)$ must be such that
\be
\epsilon_n>\epsilon_0,
\label{14}\en
for all $n>0$. This is because, from the second equality in (\ref{13}), we have
$$
\left<\hat e_n,c^\dagger \,c\hat e_n\right>=\|c\hat e_n\|^2=(\epsilon_n-\epsilon_0)\|\hat e_{n}\|^2,
$$
for all $n\geq1$, which is only compatible, if $\hat e_n\neq0$, with (\ref{14}). This, of course, imposes some constraint on the characteristic function $f(x)$, for instance that $f(x)$ is strictly increasing, as we will assume from now on. This also guarantees that each $\epsilon_n$ is not degenerate. Now, it is not hard to compute the right normalization to produce an orthonormal (o.n) basis of eigenvectors of $H$. They are
\be
e_n=\frac{1}{\sqrt{(\epsilon_n-\epsilon_0)!}}\,\hat e_n=\frac{1}{\sqrt{(\epsilon_n-\epsilon_0)!}}\,(c^\dagger)^n\hat e_0,
\label{15}\en
where $(\epsilon_n-\epsilon_0)!=(\epsilon_n-\epsilon_0)(\epsilon_{n-1}-\epsilon_0)\cdots(\epsilon_1-\epsilon_0)$, $n\geq1$ and $0! =1$. These vectors satisfy the following relations:
\be
\left\{
    \begin{array}{ll}
c^\dagger e_n=\sqrt{\epsilon_{n+1}-\epsilon_0}\,\,e_{n+1},\\
c e_n=\sqrt{\epsilon_{n}-\epsilon_0}\,\,e_{n-1},
\end{array}
        \right.
\label{16}\en
which together imply that
\be
cc^\dagger e_n=\left(\epsilon_{n+1}-\epsilon_0\right)\,e_{n}, \qquad c^\dagger c \,e_n=\left(\epsilon_{n}-\epsilon_0\right)\,e_{n},
\label{17}\en
for all $n\geq0$. In particular these imply that the original Hamiltonian $H$ can be written as $H=c^\dagger c+\epsilon_0\1$, so that $H-\epsilon_0\1$ is factorizable and coincides with $c^\dagger c$. This suggests to introduce a second operator constructed using $c$ and its adjoint, the so-called SUSY-Hamiltonian $H_{Susy}=cc^\dagger+\epsilon_0\1$. The eigenvectors of $H_{Susy}$ are the same as those of $H$, while its eigenvalues are simply shifted:
\be
H_{Susy}e_n=\epsilon_{n+1}e_n,
\label{18}\en
$n\geq0$. $H_{Susy}$ plays also a role in the determination of the characteristic function of the GHA. In fact, since $H_{Susy}=[c,c^\dagger]+c^\dagger c+\epsilon_0\1=f(H)-H+H=f(H)$, they are really the same object.

\section{Deformed GHA}\label{sectII}

\subsection{General considerations}\label{sectII_1}

In this section we will show how the general structure considered in Section \ref{sectAITGHA} can be extended in order to include Hamiltonian operators which are not necessarily self-adjoint.

The starting point of our analysis is the following definition:
\begin{defn}
\label{def1}
Let $a$, $b$ and $h$ be three operators defined on some dense domain $\D$ of the Hilbert space $\Hil$. Let us assume that $\D$ is stable under their action and under the action of their adjoints. We say that $(a,b,h)$ are {\em compatible} if, for all $\varphi\in\D$, the following equalities hold:
\be
hb\,\varphi=bf(h)\,\varphi, \qquad ah\,\varphi=f(h)a\,\varphi,
\label{21}
\en
for some fixed, strictly increasing, function $f(x)$.
\end{defn}

\vspace{2mm}

{\bf Remarks:--} (1) The first remark is that, if $h$ is bounded and $f(x)$ admits an expansion in power series $f(x)=\sum_{n=0}^\infty c_n x^n$ convergent inside an interval $|x|<M$, then $f(h)$ can be defined as $f(h)=\sum_{n=0}^\infty c_n h^n$, at least if $\|h\|<M$. On the other hand, if $h$ is unbounded but self-adjoint, $h=h^\dagger$, $f(h)$ can be defined by means of the spectral theorem. If $h$ is unbounded and not self-adjoint, defining $f(h)$ can be more complicated, but still it can be done in some cases, for instance if $h$ is similar to another operator $H$, $H=H^\dagger$, at least if the similarity map is given by a bounded operator with bounded inverse. In the most general case we can define $f(h)$ via its action on a basis of $\Hil$. Of course, the natural choice of basis would be, when possible, the set of eigenstates of $h$.

(2) In principle there is no reason a priori for taking a single function $f(x)$ in both equalities in (\ref{21}). For instance, we could assume that $hb\varphi=bf_b(h)\varphi$ and $ah\varphi=f_a(h)a\varphi$ for two different functions $f_a(x)$ and $f_b(x)$. However, in view of what we will do next, this would not be a useful choice. Moreover, the reason why we have assumed here that $f(x)$ is strictly increasing is because of what was discussed in the previous section on $\epsilon_n-\epsilon_0$, and because we want to avoid some of the eigenvalues of $h$ to be degenerate.

(3) If $h=h^\dagger$ and $a=b^\dagger$ the two equalities in (\ref{21}) collapse into a single one, which is the one considered in \cite{curado1,curado2} and reviewed in Section \ref{sectAITGHA}. Otherwise, they are different. Interestingly, they can be seen as two different intertwining relations between $h$ and $f(h)$, one due to $a$ and the other to $b$. In principle, these two equations are really independent, since $a$ and $b$ are unrelated, so far. However, in the following, see (\ref{25}), they are assumed to satisfy a suitable commutation rule, so that they are, in fact, connected. Of course, with this in mind, the whole machinery of intertwining relations, see \cite{intop,bagint1,bagint2}, could be considered in connection with our operators. However, this analysis is not particularly relevant for us now and it is postponed to a future paper.

\vspace{3mm}

\begin{prop}\label{prop1}
Let $(a,b,h)$ be compatible operators and let $\epsilon_0$ be a fixed (non-negative) real number. Let us call $\epsilon_{n+1}=f(\epsilon_n)$, $n\geq0$. Suppose now that two non-zero vectors, $\varphi_0$ and $\psi_0$, in $\D$ do exist such that $h\varphi_0=\epsilon_0\varphi_0$ and $h^\dagger \psi_0=\epsilon_0\psi_0$. Let us call
\be
\varphi_n=\frac{1}{\sqrt{(\epsilon_n-\epsilon_0)!}}\,b^n\varphi_0\qquad \psi_n=\frac{1}{\sqrt{(\epsilon_n-\epsilon_0)!}}\,(a^\dagger)^n\psi_0,
\label{22}\en
for all $n\geq0$,  where  $(\epsilon_n-\epsilon_0)!$ is defined below Eq. (\ref{15}).
%Here $0!=1$, while $(\epsilon_n-\epsilon_0)!=(\epsilon_n-\epsilon_0)(\epsilon_{n-1}-\epsilon_0)\cdots(\epsilon_1-\epsilon_0)$.
Then:

\be
h\varphi_n=\epsilon_n\varphi_n,\qquad h^\dagger \psi_n=\epsilon_n\psi_n,
\label{23}\en
for $n\geq0$. Moreover,
\be
\left<\psi_n,\varphi_m\right>=\delta_{n,m}\left<\psi_n,\varphi_n\right>.
\label{24}\en

\end{prop}

\begin{proof}

Formulas in (\ref{23}) can be proved by induction using (\ref{21}). For instance, for $n=1$ we have $$
h\varphi_1=\frac{1}{\sqrt{(\epsilon_1-\epsilon_0)!}}\,h\,b\varphi_0=\frac{1}{\sqrt{(\epsilon_1-\epsilon_0)!}}\,b\,f(h)\varphi_0=f(\epsilon_0)
\frac{1}{\sqrt{(\epsilon_1-\epsilon_0)!}}\,b\,\varphi_0=\epsilon_1\varphi_1.
$$
Let us now assume that $h\varphi_n=\epsilon_n\varphi_n$ for a fixed $n$. To prove that $h\varphi_{n+1}=\epsilon_{n+1}\varphi_{n+1}$ we observe that
$$
h\varphi_{n+1}=\frac{1}{\sqrt{(\epsilon_{n+1}-\epsilon_0)}}\,h\,b\varphi_n=\frac{1}{\sqrt{(\epsilon_{n+1}-\epsilon_0)}}\,b\,f(h)\varphi_n=
f(\epsilon_n)
\frac{1}{\sqrt{(\epsilon_{n+1}-\epsilon_0)!}}\,b\,\varphi_n=\epsilon_{n+1}\varphi_{n+1}.
$$
A similar proof holds for the vectors $\psi_n$, by using the adjoint of the equality $ah=f(h)a$. Formula (\ref{24}) is a simple consequence of the eigenvalues equations in (\ref{23}), and of the fact that the various $\epsilon_n$ are all different.

\end{proof}

\vspace{2mm}

{\bf Remarks:--} (1) Due to the stability of $\D$ it is clear that all the vectors $\varphi_n$ and $\psi_n$ belong to $\D$.

(2) For all $n\geq1$ the quantities $\epsilon_n-\epsilon_0$ are strictly positive, and $\{\epsilon_n\}$ is a strictly increasing sequence.

(3) The above result could be generalized to the situation in which the eigenvalues of $h$ and of $h^\dagger$ are different. However, doing so, we would lose the isospectrality of these two operators, which on the other hand we prefer to keep, since it has very useful consequences. {This kind of generalization has been discussed, for instance, in \cite{bagcomplex}, in finite
dimensional Hilbert spaces.}

\vspace{3mm}

Now, formulas (\ref{22}) show that $b$ and $a^\dagger$ are raising operators for the vectors in $\F_\varphi=\{\varphi_n, \,n\geq0\}$ and $\F_\psi=\{\psi_n, \,n\geq0\}$ respectively. We expect that $a$ and $b^\dagger$ are lowering operators. However, this is not true in general, but it is true if the following commutation rule between $a$ and $b$ is satisfied:
\be
[a,b]=f(h)-h,
\label{25}\en
at least on $\D$. Of course, this also implies that $[b^\dagger,a^\dagger]=f(h^\dagger)-h^\dagger$ on $\D$. The following lowering equations can now easily be proved, if $a\varphi_0=b^\dagger\psi_0=0$\footnote{As in Proposition \ref{prop0} these equalities are surely true if the sets $\F_\varphi$ and $\F_\Psi$ are complete.}:
\be
a\varphi_n=\sqrt{\epsilon_n-\epsilon_0}\,\varphi_{n-1},\qquad b^\dagger\psi_n=\sqrt{\epsilon_n-\epsilon_0}\,\psi_{n-1},
\label{26}\en
for all $n\geq1$. Also, we get
\be
\left\{
    \begin{array}{ll}
ba\varphi_n=(\epsilon_n-\epsilon_0)\varphi_n,\qquad\quad ab\varphi_n=(\epsilon_{n+1}-\epsilon_0)\varphi_n\\
a^\dagger b^\dagger\psi_n=(\epsilon_n-\epsilon_0)\psi_n,\qquad\, b^\dagger a^\dagger\psi_n=(\epsilon_{n+1}-\epsilon_0)\psi_n,
\end{array}
        \right.
\label{27}\en
which show that $\varphi_n$ is an eigenstate of both $ba$ and $ab$, while each $\psi_n$ is eigenstate of both $a^\dagger b^\dagger$ and $b^\dagger a^\dagger$.
Now, under the assumptions of Proposition \ref{prop1}, if $\left<\psi_0,\varphi_0\right>=1$, it is a standard computation to check that  $\left<\psi_n,\varphi_m\right>=\delta_{n,m}$.

\begin{defn}\label{defDGHA}

The set of commutation rules in (\ref{21}) and (\ref{25}) satisfied by the operators $a$, $b$ and $h$ define a {\em deformed generalized Heisenberg algebra} (DGHA).

\end{defn}

\vspace{2mm}

{\bf Remark:--} In view of what we have seen in Section \ref{sectAITGHA}, and of the eigenvalue equations in (\ref{27}), two SUSY Hamiltonians can be introduced for $h$ and $h^\dagger$, and their eigenvalues and eigenvectors can be easily be deduced out of those above.

\subsection{Constructing a deformed GHA from a GHA }
\label{DGHA}

Let $c$ and $H$ be operators satisfying the GHA as discussed in Section \ref{sectI}. As we have seen, they satisfy the following:
$$
cH=f(H)c,\qquad [c,c^\dagger]=f(H)-H,
$$
where $H=H^\dagger$ and $f(x)$ is a strictly increasing function. If we call $e_0$ the ground state of $H$, we require that $ce_0=0$ and that $e_0\in\D$, a suitable dense subset of $\Hil$. Let now $S$ be an invertible operator which leaves $\D$ stable, together with $S^{-1}$. Then, introducing $a=ScS^{-1}$, $b=Sc^\dagger S^{-1}$, $\varphi_0=Se_0$, $\psi_0=(S^{-1})^\dagger e_0$ and $h=SHS^{-1}$, it is clear that these new operators and vectors satisfy all properties required in Section \ref{sectII_1}. In particular, (\ref{21}) and (\ref{25}) are satisfied. Hence two biorthogonal sets can be constructed as in (\ref{22}), and these are eigenstates of $h$ and $h^\dagger$. Notice that, as it is discussed extensively in \cite{baginbagbook}, these two sets in general are not bases for the Hilbert spaces, even if they turn out, quite often in concrete examples, to be complete. However, if both $S$ and $S^{-1}$ are bounded, then $\F_\varphi$ and $\F_\psi$ are biorthogonal Riesz bases. Otherwise they are $\D$-quasi bases, see Appendix and reference \cite{baginbagbook}.

\vspace{2mm}

These simple steps show how a GHA can be modified in order to get a DGHA. In a certain sense, this result can be inverted: under natural assumptions, any DGHA gives rise to a GHA. Let us consider the operators $(a,b,h)$ satisfying Definition \ref{defDGHA}. To simplify the situation, we assume here that the sets $\F_\varphi$ and $\F_\psi$ constructed as discussed in Proposition \ref{prop1} are biorthogonal Riesz bases. Then the operators $S_\varphi f:=\sum_{n}\left<\varphi_n,f\right>\varphi_n$ and $S_\psi g:=\sum_{n}\left<\psi_n,g\right>\psi_n$ can be defined in all of $\Hil$. Moreover, because of our assumption, a bounded operator $R$ exists, with bounded inverse $R^{-1}$, and an o.n. basis $\F_v=\{v_n\}$ such that $\varphi_n=Rv_n$ and $\psi_n=(R^{-1})^\dagger v_n$, for all $n$. Then we easily deduce that $S_\varphi=RR^\dagger$, while $S_\psi=S_\varphi^{-1}$. Hence $S_\varphi$ and $S_\psi$ are both bounded and positive. Hence their (unique) positive square roots exist. If, for simplicity, we assume that also $S_\varphi^{1/2}$, $S_\psi^{1/2}$ leave $\D$ invariant, then the operator $c=S_\psi^{1/2}aS_\varphi^{1/2}$ also maps $\D$ in $\D$. This surely happens if $\D=\Lc_\varphi\cap \Lc_\psi$, where $\Lc_\varphi$ is the linear span of the $\varphi_n$'s and $\Lc_\psi$ is the linear span of the $\psi_n$'s. Notice that both $\Lc_\varphi$ and $\Lc_\psi$ are dense in $\Hil$,since $\F_\varphi$ and $\F_\psi$ are Riesz bases for $\Hil$, and we are here assuming that their intersection is dense as well. The adjoint of $c$, $c^\dagger$, turns out to be $c^\dagger=S_\psi^{1/2}bS_\varphi^{1/2}$ on $\D$. Now, if we introduce a new operator $H$ on $\D$ as $Hg=S_\psi^{1/2}hS_\varphi^{1/2}g$, $g\in\D$, and new vectors $e_n=S_\varphi^{1/2}\psi_n$, we go back to what discussed in Section \ref{sectAITGHA}.

\section{Some classical examples}\label{sectIII}

We now discuss some classical examples which fit our assumptions,
and we will deform them according with what discussed in Section \ref{DGHA}.

\subsection{P\"oschl-Teller potentials}

As a concrete example of the general scheme presented above, we first consider the quantum model of  a one-dimensional particle subjected to the symmetric P\"oschl-Teller potential
$$
V_{\lambda}(x)=\frac{\lambda(\lambda-1)}{\sin^2x}\, ,
$$
where $\lambda\geq1$ and $x\in (0,\pi)$. For $\lambda > 1$, this potential is a regularization of the infinite square well ($\lambda =1$) and extrapolates both the latter and the harmonic oscillator (for small $\vert x-\pi/2\vert$). The Hamiltonian of the particle is, fixing
$\hbar=2m=1$, $$H_{\lambda}=-\frac{\ud^2}{\ud x^2}+V_{\lambda}(x)\, , $$
 and the eigenvalue equation for $H_{\lambda}$ can be explicitly solved, (see for instance \cite{antoine_etal_01} and references therein):
\be
H_{\lambda}\,e^{\lambda}_n(x)=\epsilon^{\lambda}_n e^{\lambda}_n(x)\, ,
\label{28}\en
where $\epsilon^{\lambda}_n=(n+\lambda)^2$ and
\be
e^{\lambda}_n(x)=K^{\lambda}_n\,\sin^{\lambda} x\,\mathrm{C}_n^{\lambda}\left(\cos x\right)\,.
\label{29}\en
Here
$$
K^{\lambda}_n=\Gamma(\lambda)\frac{2^{\lambda-1/2}}{\sqrt{\pi}}\sqrt{\frac{n!(n+\lambda)}{\Gamma(n+2\lambda)}}
$$
is a normalization constant and $\mathrm{C}_n^{\lambda}$ is the  Gegenbauer polynomial of degree $n$ \cite{magnus66}. The set $\{e^{\lambda}_n(x)\}$ is an orthonormal basis.
We can check that $\epsilon^{\lambda}_{n+1}=(\sqrt{\epsilon^{\lambda}_n}+1)^2$, so that $f(x)=(\sqrt{x}+1)^2$ is the characteristic function for the system.

\vspace{2mm}

{\bf Remark:--}
For the Hamiltonian $H_{\lambda}$  a SUSY approach has been discussed in \cite{bergasiyou10,bersiyou12}, with corresponding lowering and raising operators, $A_{\lambda}= \dfrac{\ud}{\ud x} -\lambda \cot x$ and  $A^{\dag}_{\lambda}= -\dfrac{\ud}{\ud x} -\lambda \cot x$. In this case one finds the Darboux factorization $H_{\lambda}= A^{\dag}_{\lambda} A_{\lambda} +\epsilon^{\lambda}_0$, and one gets $H_{\lambda+1}=A_{\lambda}\, A^{\dag}_{\lambda} + \epsilon^{\lambda}_0$ for its partner. However, this representation of $H_\lambda$ does not satisfy the assumptions of a GHA since  $A_{\lambda}$ and  $A^{\dag}_{\lambda}$ are not ladder operators. In fact, they shift both the polynomial degree and the parameter $\lambda$ as
$$
A_{\lambda}\, e^{\lambda}_n(x) = \sqrt{\epsilon^{\lambda}_n-\epsilon^{\lambda}_0 }\, e^{\lambda+1}_{n-1}(x)\, , \quad A^{\dag}_{\lambda}\, e^{\lambda+1}_n(x) = \sqrt{\epsilon^{\lambda}_{n+1}-\epsilon^{\lambda}_0 }\, e^{\lambda}_{n+1}(x)\,.
$$
as expected from the  SUSY quantum mechanics formalism, which is not what should be satisfied. In fact, to be relevant for our purposes, we should find a factorization which leaves $\lambda$ unchanged, while changing $n$.  Hence, let us introduce the following ladder operators,
\begin{equation}
\label{Bnl}
B_{\lambda}= -\sin x \frac{\ud}{\ud x}+\cos x\,(\hat{N}_{\lambda}+\lambda)\, , \quad B^{\dag}_{\lambda}= \sin x \frac{\ud}{\ud x}+(\hat{N}_{\lambda}+\lambda +1)\,\cos x\, ,
\end{equation}
where the diagonal ``number'' operator $\hat{N}_{\lambda}$ is defined by its action on the basis \eqref{29} as
\begin{equation}
\label{Nl}
\hat{N}_{\lambda}\, e^{\lambda}_n(x)= n e^{\lambda}_n(x)\,.
\end{equation}
The actions of the operators \eqref{Bnl} on the basis are easily derived from \eqref{uEn} and \eqref{u2dEn} by putting $u = \cos x$:
\begin{align}
\label{BEn}
B_{\lambda} \, \mathcal{E}^{\lambda}_{n}(u)&=\sqrt{\frac{n(n+\lambda)(n+2\lambda-1)}{n-1+\lambda}} \mathcal{E}^{\lambda}_{n-1}(u)\,,\\\label{BdEn} B^{\dag}_{\lambda} \, \mathcal{E}^{\lambda}_{n}(u)&= \sqrt{\frac{(n+1)(n+1+\lambda)(n+2\lambda)}{n+\lambda}} \mathcal{E}^{\lambda}_{n+1}(u)\, .
\end{align}
The next step is to build the operator $c$ corresponding to the Hamiltonian $H_{\lambda}$. We first derive from  \eqref{BEn} the diagonal operator
\begin{equation}
\label{BdB}
B^{\dag}_{\lambda}\,B_{\lambda} = \frac{\hat{N}_{\lambda}(\hat{N}_{\lambda}+ \lambda)(\hat{N}_{\lambda}+\lambda-1)}{\hat{N}_{\lambda} -1 +\lambda}= \left(G_{\lambda}\left(\hat{N}_{\lambda}\right)\right)^{-2}\, \left(H_{\lambda} -\lambda^2\right)\, ,
\end{equation}
where, for $\lambda >1$,  the strictly increasing positive bounded function $G_{\lambda}(t)$ is defined by
\begin{equation}
\label{Gx}
G_{\lambda}(t)= \sqrt{\frac{(t+2\lambda)(t-1+\lambda)}{(t+2\lambda -1)(t+\lambda)}}\, , \quad 0 <\sqrt{\frac{2(\lambda-1)}{2\lambda-1}} \leq G_{\lambda}(t) < 1
\end{equation}
The limit case of the infinite square well, for which $\lambda =1$, deserves a particular treatment and will be examined in the sequel.
We can rewrite $G_{\lambda} (t)$ as:
\begin{equation}
\label{GT}
G_{\lambda}(t)= T_\lambda^{1/2}(t-1) T_\lambda^{-1/2}(t) \, ,
\end{equation}
where
\begin{equation}
\label{Tlambda}
T_{\lambda}(t)= \frac{t+ \lambda }{t+2 \lambda } \, ,
\end{equation}
which is also positive bounded with bounded inverse and it is a
monotonically increasing function, from $1/2$ ($t=0$) to $1$ (t $\to \infty$).

We can now introduce our  operators corresponding to $c$ and $c^{\dag}$,
\begin{align}
\label{cnl}
C_{\lambda}=& B_{\lambda} \,G_{\lambda}(\hat{N}_{\lambda}) = T_{\lambda}^{1/2}(t) B_{\lambda}  T_{\lambda}^{-1/2}(t) \\
C^\dagger_{\lambda}=&G_{\lambda}(\hat{N}_{\lambda})\, B^{\dag}_{\lambda} =T_{\lambda}^{-1/2}(t) B_{\lambda}^\dagger  T_{\lambda}^{1/2}(t) \, ,
\end{align}
where we have used the commutation relations $\mathcal{G}(\hat{N}+1) B_\lambda = B_\lambda \mathcal{G}(\hat{N})$, $\mathcal{G}(\hat{N}-1) B_\lambda^\dagger = B_\lambda^\dagger \mathcal{G}(\hat{N})$,
valid for any smooth function  $\mathcal{G}(\hat{N})$.

It is easy to check that, as simplified  ladder operators, they  obey all expected GHA properties for this example of potential. In particular:
\begin{align}
\label{Clad}
C_{\lambda}\, \mathcal{E}^{\lambda}_{n}(u)&= \sqrt{n(n+2\lambda)}\, \mathcal{E}^{\lambda}_{n-1}(u)= \sqrt{\epsilon_n^{\lambda} - \epsilon_0^{\lambda}}\,\, \mathcal{E}^{\lambda}_{n-1}(u)\, , \\
\label{Clad}
C^{\dag}_{\lambda}\, \mathcal{E}^{\lambda}_{n}(u)&= \sqrt{(n+1)(n+1+2\lambda)}\, \mathcal{E}^{\lambda}_{n+1}(u)= \sqrt{\epsilon_{n+1}^{\lambda} - \epsilon_0^{\lambda}}\,\, \mathcal{E}^{\lambda}_{n+1}(u)\, .
\end{align}
\begin{equation}
\label{CLH}
C_{\lambda}\, H_{\lambda}= f\left(H_{\lambda}\right)\,C_{\lambda} \, , \ \mbox{with} \ f\left(H_{\lambda}\right)= \left(\sqrt{H_{\lambda}} + 1\right)^2\,.
\end{equation}
\begin{equation}
\label{ }
C^{\dag}_{\lambda}\, C_{\lambda}= H_{\lambda}- \epsilon_0^{\lambda}\,I\, , \quad  C_{\lambda}\,C^{\dag}_{\lambda}= f\left(H_{\lambda}\right) - \epsilon_0^{\lambda}\,I\,,\quad \left[C_{\lambda}\,, \,C^{\dag}_{\lambda}\right]= f\left(H_{\lambda}\right) - H_{\lambda}\,.
\end{equation}
A similar physical realization of the P\"oschl-Teller potential was also recently realized, see \cite{RMECLRpra2017}. There, the P\"oschl-Teller creation and annihilation operators are written in a different way, but they are completely equivalent to our $C_\lambda$ and $C_\lambda^\dagger $ operators.

\subsubsection{Deforming P\"oschl-Teller}

If we now want to produce a DGHA, the easiest procedure consists in fixing a bounded operator with bounded inverse, and working as proposed at the beginning of Section \ref{DGHA}. For that let us illustrate the procedure by picking the following function of $x$,
% $S(x)=\dfrac{1+2x}{1+x}$.
 $ S(x) = (1+2x)/(1+x)$.
It is clear that this  can be considered as a bounded multiplication operator, with bounded inverse, for all $x\in [0,\pi]$. Hence we can use it and $S^{-1}(x)$ to deform the system, similarly to what we have done before: $a_\lambda=S(x)C_\lambda S^{-1}(x)$, $b_\lambda=S(x)C_\lambda^\dagger S^{-1}(x)$, $\varphi_0^\lambda(x)=S(x)e_0^\lambda(x)$ and $\psi_0^\lambda(x)=S^{-1}(x)e_0^\lambda(x)$. In particular, the Hamiltonian takes the following explicit expression:
\be
h_\lambda=S(x)H_\lambda S^{-1}(x)=-\frac{\ud^2}{\ud x^2}+\frac{2}{(1+x)(1+2x)}\,\frac{\ud}{\ud x}+V_\lambda(x)-\frac{4}{(1+x)(1+2x)^2},
\label{ptham}\en
which is manifestly non self-adjoint. This describe an Hamiltonian with a new potential $V_\lambda(x)-\frac{4}{(1+x)(1+2x)^2}$, plus a term which is proportional to a given function of $x$ multiplying the momentum operator. For this Hamiltonian, and its adjoint $h_\lambda^\dagger$, the eigenvectors can be deduced as discussed in Section \ref{sectII_1}.

\subsection{The infinite square well}

As already stated, the infinite square well can be viewed as a particular case of the P\"oschl-Teller potentials, with  $\lambda=1$. So, in principle, we just adapt our previous results to this particular situation simply by fixing this value for $\lambda$. However,  we follow  here a slightly different choice for  the operator $S$ used to deform the GHA.
We also simplify our notations by dropping the subscript ``1", which is proper to the
infinite square well.  Hence, we now consider the Hamiltonian of an infinite well in the interval $[0,\pi]$:
\be
H=-\,\frac{\ud^2}{\ud x^2} + V(x)\, ,
\label{iw1}
\en
 where
$$
V(x)=\left\{
    \begin{array}{ll}
0\qquad\qquad\,\, x\in(0,\pi)\, , \\
\infty\qquad\qquad \mbox{elsewhere}
\end{array}
        \right.\, .
$$
Eigenvalues and corresponding eigenfunctions read $\epsilon_n=(n+1)^2$ and $e_n(x)=\sqrt{\frac{2}{\pi}}\,\sin(n+1)x$ respectively. Hence $He_n=\epsilon_ne_n$, $n\geq0$, and the characteristic function $f(x)$ for the GHA is as before: $f(x)=(\sqrt{x}+1)^2$. The function $G(t)$  becomes
\begin{equation*}
G(t)= \frac{\sqrt{t\,(t+2)}}{t+1}\, ,
\end{equation*}
i.e., the ratio of the geometric mean of $t$ and $t+2$  to the arithmetic one.
The ladder operators \eqref{cnl} assume their simplest expression:
\begin{align}
\label{iw2}
    C &= B \, G(\hat N) = T^{1/2}(\hat N)  \, B \, T^{-1/2}(\hat N) \\
    C^\dagger &=  G(\hat N) \, B^\dagger = T^{-1/2}(\hat N)  \, B^\dagger \, T^{1/2}(\hat N)\, ,
    \end{align}
with $T(t)= (t+1)/(t+2)$ in agreement with Eq. (\ref{Tlambda}) for $\lambda =1$.
$\hat N$ is still the usual number operator $\hat N e_n(x)=ne_n(x)$, i.e., $\hat N = \sqrt{H}-1$.

\subsubsection{Deforming the infinite square-well potential}

What we want to do here is to deform this example in the way we have proposed in Section \ref{DGHA}: for that, let $\sigma(x)$ be a real function, which is supposed to satisfy the following bounds: $0<\sigma_m\leq\sigma(x)\leq\sigma_M<\infty$, for almost all $x\in[0,\pi]$. This means that the inverse of $\sigma(x)$ exists. Moreover, using the spectral theorem, we can define $\sigma\left(\hat N +1\right)$ and this operator commutes with $H$. Therefore, if we take $S:=\sigma\left(\hat N +1\right)$, it follows that $h=H$. Moreover, $\varphi_0=Se_0=\sigma(1)e_0$, while $\Psi_0=(S^{-1})^\dagger e_0=(\sigma(1))^{-1}e_0$. Hence, $\varphi_0$ and $\Psi_0$ differ from $e_0$ only for a normalization, and they satisfy $\left<\varphi_0,\Psi_0\right>=1$. The operators $a$ and $b$ act on $e_n(x)$ in a (formally) easy way. This follows from the fact that $Se_n=\sigma\left(\hat N+1\right)e_n=\sigma(n+1)e_n$. Then, for instance,
$$
b\,e_n=SC^\dagger S^{-1}e_n=\sigma\left(\hat N +1\right) C^\dagger \sigma\left(\hat N+1\right)^{-1}e_n=(\sigma(n+1))^{-1}\sigma\left(\hat N+1\right) C^\dagger e_n=$$ \be=(\sigma(n+1))^{-1}\sigma\left(\hat N +1\right)\sqrt{\epsilon_{n+1}-\epsilon_0}\,e_{n+1}=
\frac{\sigma(n+2)}{\sigma(n+1)}\,\sqrt{\epsilon_{n+1}-\epsilon_0}\,e_{n+1},
\label{iw3}\en
with a similar result for $a$. We see that $b$ is still a raising operator, also with respect to the original o.n. basis $\F_e$, but with a slightly different coefficient, which involves $\sigma$.

\vspace{2mm}

A different conclusion is deduced if we consider a different choice of the operator $S$. In particular, if we take now $S$ to be the following multiplication operator: $S=(\sigma(x))^{-1}$, where $\sigma(x)$ is as before. Then $h$ turns out to be different from $H$, in general. In fact:
$$
h=-\frac{\sigma^{\prime\prime}(x)}{\sigma(x)}-\frac{2\sigma'(x)}{\sigma(x)}\,\frac{\ud}{\ud x}-\frac{\ud^2}{\ud x^2},
$$
which is manifestly non self-adjoint. Let us see what happens with the particular choice of $\sigma(x)=\alpha+\cos(k_0x)$, where $\alpha>1$ and $k_0\geq1$ is a fixed natural number. Of course we have $\sigma_m=\alpha-1>0$ and $\sigma_M=\alpha+1<\infty$. A simple computation shows that
$$
\varphi_0(x)=\sqrt{\frac{2}{\pi}}\,\frac{\sin(x)}{\alpha+\cos(k_0x)}, \qquad \Psi_0(x)=\alpha e_0(x)+\frac{1}{2}\left(e_{k_0+1}(x)-e_{k_0-1}(x)\right).
$$
We see that, while $\Psi_0(x)$ is just a linear combination of three elements of $\F_e$, $\varphi_0(x)$ is an infinite series of such elements.
The analogous of formula (\ref{iw3}) can now be deduced using the equality
$$
S^{-1}e_n(x)=\alpha e_n+\frac{1}{2}\left(e_{n+k_0}(x)+e_{n-k_0}(x)\right),
$$
which follows from some well known trigonometric identities. We restrict here to the case $n\geq k_0$. The opposite case can be easily deduced by simply using the parity properties of $e_n(x)$. We get
$$
b\,e_n=\frac{1}{\alpha+\cos(k_0x)}\left(\alpha\sqrt{\epsilon_{n+1}-\epsilon_0}\,e_{n+1}+\frac{\sqrt{\epsilon_{n+1+k_0}-\epsilon_0}}{2}e_{n+k_0+1}+
\frac{\sqrt{\epsilon_{n+1-k_0}-\epsilon_0}}{2}e_{n-k_0+1}
\right),
$$
which clearly shows how $b$ is no longer a raising operator, in this case, for the family $\{e_n\}$. The action of $a$ on $e_n$ can be deduced in a similar way. As for $h$, we get
$$
h=\frac{1}{\alpha+\cos(k_0x)}\left(k_0\cos(k_0x)+2k_0\sin(k_0x)\,\frac{\ud}{\ud x}-(\alpha+\cos(k_0x))\frac{\ud^2}{\ud x^2}\right).
$$
This operator looks extremely different from the one in (\ref{ptham}), even when $\lambda$ in $h_\lambda$ is fixed to be one. This is a consequence of the two different choices of the similarity operator $S$ in this case, and in (\ref{ptham}).

\subsection{The harmonic oscillator}

\subsubsection{A familiar preliminary}
Let $c$ be the standard bosonic lowering operator, satisfying the canonical commutation rule $[c,c^\dagger]=\1$, and let $H_0=c^\dagger c$. Of course, $H_0=H_0^\dagger$. Working in the coordinate representation, $e_0(x)=\frac{1}{\pi^{1/4}}e^{-x^2/2}$, $x \in (-\infty,\infty)$, is a function annihilated by $c$: $ce_0=0$. Now, taking $f(x)=x+1$ and identifying $\D$ with the set of test functions $\Sc(\Bbb R)$, a DGHA is trivially recovered if $a=c$, $b=c^\dagger$ and $h=H_0$. In this case, clearly, $\varphi_0(x)=\psi_0(x)=e_0(x)$, and $\D$ is stable under the action of $c$ and $c^\dagger$. Also, $f(x)$ is strictly increasing.

The physical realization of the operators $c = (1/\sqrt{2}) (x+ i d/dx)$ and its adjoint $c^\dagger$ is well-known, of course,  and they satisfy the Weyl-Heisenberg algebra.

\subsubsection{Deforming the harmonic oscillator}

In order to implement a non self-adjoint deformation of the harmonic oscillator, according to
Section \ref{DGHA}, we have to find a positive, bounded, with a bounded inverse, multiplication operator $S(x)$. As the variable $x \in (-\infty,\infty)$, a function like $S(x) =2+ \tanh(x)$ satisfy all necessary requirements of Section \ref{DGHA}. This function increases, monotonically, from $1$ ($ t \to -\infty$) up to $3$ ($t \to \infty$).
Using $S(x)$  and its inverse, $S^{-1}(x)$, the harmonic oscillator can be deformed in the following way:
$a = S(x) \, c \,  S^{-1}(x)$, $b = S(x) \, c^\dagger \, S^{-1}(x)$, $\varphi_0(x) = S(x) e_0(x)$ and $\psi_0(x) =
S^{-1}(x) e_0(x)$. The new  Hamiltonian takes the form:
\begin{equation}
\label{hhod}
h = S(x) \, H_0 \, S^{-1}(x) = -\frac{\ud^2}{\ud x^2}+ 2 (1-\tanh(x)) \,\frac{\ud}{\ud x} - 2 (1-\tanh(x)) + \frac{x^2}{2} \, ,
\end{equation}
which is clearly non self-adjoint. This  system has an effective potential $V_{eff}(x) =\frac{x^2}{2} - 2 (1-\tanh(x)) $, shown in Fig \ref{potentialho}, that is non symmetric and is slightly displaced from the origin. There is also a term that is a function of $x$ multiplied by the derivative  $\ud / \ud x$, which seems to appear any time we use a function of $x$ to deform the GHA.

\begin{figure}
\begin{center}
\includegraphics[width=3in]{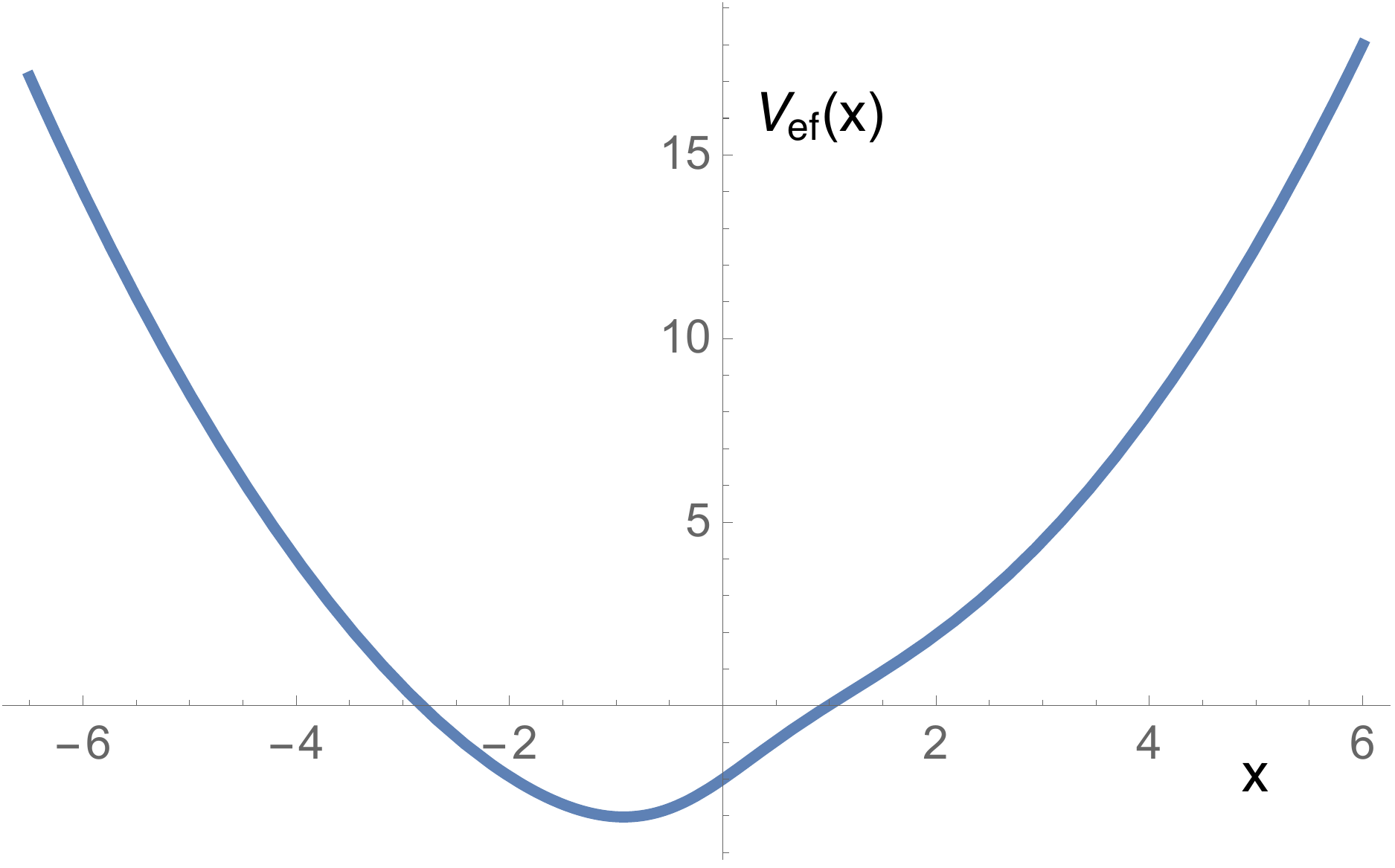}
\caption{Effective potential of a deformed non self-adjoint harmonic oscillator, with $S(x) = 2 + \tanh(x)$. }
\label{potentialho}
\end{center}
\end{figure}

\section{Relations with nonlinear pseudo-bosons}\label{sectIV}

\subsection{$\D$-pseudo bosons}

Let $A$ and $B$ be pseudo-bosonic operators in the sense of \cite{baginbagbook}, i.e., $[A,B]= \1$ and $B$ is supposed to be not equal to $A^{\dag}$,  and let $\D$ be the dense domain left stable by these operators and by their adjoints. In this situation the assumptions in Section \ref{sectII_1} are satisfied if we take $h=BA$ and $f(x)=x+1$, as for the harmonic oscillator. The sets of eigenvectors $\F_\varphi$ and $\F_\psi$ can be biorthogonal (Riesz) bases, or $\D$-quasi bases, see Appendix, depending on the explicit details of the original pair of operators $(A,B)$, as widely discussed in \cite{baginbagbook} and in references therein.

\subsection{Examples related to $\D$-pseudo bosons}

The pseudo-bosonic operators $A$ and $B$ used in the previous example can be used to construct a new class of examples. For that we define new operators
$$
a=A, \quad b=N_0^kB,
$$
where $N_0=BA$ and $k$ is a fixed positive integer, $k=1,2,3,\ldots$. We further define $N=ba=N_0^{b+1}$ and $h=N=N_0^{k+1}$. It is possible to check that, putting $f(h)=[a,b]+h$, we have $f(h)=(N_0+\1)^{k+1}$, for all fixed $k$. Of course, since the eigenvalues of $N_0$ are just the natural numbers, including zero, $f(h)$ is positive and increasing. The equalities in (\ref{21}) can be further checked explicitly, as a consequence of the following equalities:
$$
AN_0^k=(N_0+\1)^kA,\qquad BN_0^k=(N_0-\1)^kA, \qquad N_0^kB=B(N_0+\1)^k,
$$
for all $k=0,1,2,\ldots$.

Few years ago, \cite{fbnlpbs}, the concept of PBs was generalized to consider quantum systems in which the eigenvalues of the Hamiltonian (and its adjoint) do not depend linearly on the quantum number labeling the eigenstates. Few examples were discussed in \cite{fbnlpbs}, and in other and more recent papers. What we will show now is that there is a strong connection between these nonlinear pseudo-bosons (NLPBs) and the DGHA discussed in Section \ref{sectII}. To show that, we briefly recall how NLPBs are constructed.

Let us consider a strictly increasing sequence $\{\epsilon_n\}$: $0=\epsilon_0<\epsilon_1<\cdots<\epsilon_n<\cdots$. Further, let us consider two
operators $A$ and $B$ on $\Hil$, and let us suppose that there exists a set $\D\subset\Hil$ which is dense in $\Hil$, and which is stable under the action of $A, B$ and their adjoints.

\begin{defn}\label{def2}
We will say that the triple $(A,B,\{\epsilon_n\})$ is a family of $\D$-non linear pseudo-bosons ($\D$-NLPBs) if the following properties hold:
\begin{itemize}

\item {\bf p1.} a non zero vector $\Phi_0$ exists in $\D$ such that $A\,\Phi_0=0$;

\item {\bf  p2.} a non zero vector $\eta_0$ exists in $\D$ such that $B^\dagger\,\eta_0=0$;

\item {\bf { p3}.} Calling
\be \Phi_n:=\frac{1}{\sqrt{\epsilon_n!}}\,B^n\,\Phi_0,\qquad \eta_n:=\frac{1}{\sqrt{\epsilon_n!}}\,{A^\dagger}^n\,\eta_0, \label{55} \en we
have, for all $n\geq0$, \be A\,\Phi_n=\sqrt{\epsilon_n}\,\Phi_{n-1},\qquad B^\dagger\eta_n=\sqrt{\epsilon_n}\,\eta_{n-1}. \label{56}\en
\item {\bf { p4}.} The set $\F_\Phi=\{\Phi_n,\,n\geq0\}$ is a basis for $\Hil$.

\end{itemize}

\end{defn}

Of course, since $\D$ is stable under the action of $B$ and $A^\dagger$, it follows that $\Phi_n, \eta_n\in \D$, for all $n\geq0$. Notice that $\D$-PBs are recovered by fixing $\epsilon_n=n$. Notice also that the set $\F_\eta=\{\eta_n,\,n\geq0\}$ is automatically a basis for $\Hil$ as well. This follows from the fact that, calling $M=BA$, we have
$M\Phi_n=\epsilon_n\Phi_n$ and $M^\dagger\eta_n=\epsilon_n\eta_n$. Therefore, choosing the normalization of $\eta_0$ and $\Phi_0$ in such a way
$\left<\eta_0,\Phi_0\right>=1$, $\F_\eta$ is biorthogonal to the basis $\F_\Phi$. Then, it is possible to check that $\F_\eta$ is the unique
basis which is biorthogonal to $\F_\Phi$. We refer to \cite{fbnlpbs} for more details.

\vspace{2mm}

To connect NLPBs to DGHA we first observe that, if we consider a DGHA with $\epsilon_0=0$, this automatically gives rise to a family of $\D$-NLPBs. For that it is sufficient to identify $(A,B,\Phi_0,\eta_0)$ in Definition \ref{def2} with the quantities $(a,b,\varphi_0,\psi_0)$ introduced in Section \ref{sectII}, respectively. In fact, with this identification, conditions {\bf p1}, {\bf p2} and {\bf p3} are surely satisfied. Moreover, $\F_\Phi$ is a basis if $\F_\varphi$ is a basis.

\vspace{2mm}

It is also possible to check that the opposite holds, at least under some further minor assumption: for that we start with a family of $\D$-NLPBs, and we identify $(a,b,h,\varphi_0,\psi_0)$ with $(A,B,BA,\Phi_0,\eta_0)$. Moreover, we identify also $f(h)$ with $AB$. Notice that $f(h)=AB=[A,B]+h$ so that (\ref{25}) is automatically satisfied, at least if $[A,B]$ can be written in terms of $h$ and the resulting $f(x)$ is increasing. This becomes, in the case of $\D$-PBs, $f(h)=h+\1$, as we have already found before. In the general case, it is easy to see that  $\left<\varphi_n,[A,B]\varphi_n\right>=(\epsilon_{n+1}-\epsilon_n)\|\varphi_n\|^2$, for all $n$. This, however, does not imply that $\left<f,[A,B]f\right>$ is automatically positive, since $\F_\varphi$ is not an o.n. basis. However it is yet a strong indication that $[A,B]$ is positive. This can be explicitly checked at least on those $f\in\D$ for which each $\left<f,\varphi_n\right>\left<\psi_n,f\right>$ is non negative, at least if $\inf_n(\epsilon_{n+1}-\epsilon_n)>0$. In fact, in this case, we have
$$
\left<f,[A,B]f\right>=\sum_n(\epsilon_{n+1}-\epsilon_n)\left<f,\varphi_n\right>\left<\psi_n,f\right>\geq \inf_n(\epsilon_{n+1}-\epsilon_n)\sum_n\left<f,\varphi_n\right>\left<\psi_n,f\right>=$$
$$=\inf_n(\epsilon_{n+1}-\epsilon_n)\|f\|^2>0,
$$
using the fact that $\F_\varphi$ and $\F_\psi$ are biorthogonal (or $\D$-quasi) bases. Finally notice that the commutation rules for DGHA in (\ref{21}) are trivially satisfied with our choices: $hb=(ba)b$, and $bf(h)=b(ab)$. Also, $ah=a(ba)$, while $f(h)a=(ab)a$, and we see that, in fact, (\ref{21}) are satisfied.

\vspace{2mm}

\subsection{Quons}

In \cite{fbnlpbs} it is discussed how quons are connected with $\D$-NLPBs. Therefore it is not a surprise that quons are connected to DGHA. Let us first consider {\em ordinary quons}, i.e. operators $c$ and $c^\dagger$ obeying the following commutation rule: $c\,c^\dagger-q\,c^\dagger\,c=\1$, where $q\in[-1,1]$. Of course, $q=-1$ gives back CAR, while if $q=1$ we recover CCR. In general, if we introduce $a=c$, $b=c^\dagger$ and $h=c^\dagger c$, it is easy to check that they give rise to a DGHA with $f(x)=qx+1$, which is increasing if $q\in]0,1]$. Therefore, under this limitation, we recover the algebraic structure discussed in Section \ref{sectI}.

The same conclusion can be found if we consider a deformed version of quons, see \cite{bagquons2}. In this case we have two operators, $a$ and $b$, with $b\neq a$, satisfying $a\,b-q\,b\,a=\1$, where $q\in[-1,1]$, and we define $h=ba$. Of course these commutation rules should be defined on a dense set, possibly stable under the action of the operators involved in the game. Once again, also in this case, it is possible to show that these operators give rise to a DGHA with the same $f(x)$ as for the ordinary quons.

\section{Conclusions}\label{sectconcl}

We have shown how GHA can be deformed using ideas borrowed from the theory of pseudo-bosons, and that, in this way, biorthogonal sets of eigenvectors of the related, non self-adjoint operators, can be explicitly constructed. This strategy has been applied to several examples, and relations with NLPBs and quons have also been described. We plan to consider more applications and construct new quantum solvable models adopting our ideas. {  In particular, it will be interesting to see what our strategy can give when taking the systems in \cite{quons} as starting points.}

\section*{Acknowledgements}
FB acknowledges support from the GNFM of Indam and from the University of Palermo. JPG acknowledges the CBPF for financial support, and EMFC acknowledges the Brazilian scientific agencies CNPq and FAPERJ for financial support.

\appendix

\numberwithin{equation}{section}

\section{Gegenbauer recurrence formulae and useful definitions}
\subsection{Gegenbauer recurrences}
The following formulae are relevant to the content of this article. They are derived from \cite{magnus66} and adapted to our needs. Let us consider in the Hilbert space $L^2\left([-1,1]\,,\, \dfrac{\ud u}{\sqrt{1-u^2}}\right)$ the orthonormal basis
\begin{equation}
\label{Enu}
\mathcal{E}^{\lambda}_{n}(u)= K_n(\lambda)\,(1-u^2)^{\lambda/2}\, \mathrm{C}_n^{\lambda}(u)\, , \quad K^{\lambda}_n(\lambda)=\Gamma(\lambda)\frac{2^{\lambda-1/2}}{\sqrt{\pi}}\sqrt{\frac{n!(n+\lambda)}{\Gamma(n+2\lambda)}}\, .
\end{equation}
The first recurrence formula concerns the self-adjoint multiplication operator $Qf(u) = u f(u)$.
\begin{equation}
\label{uEn}
u\, \mathcal{E}^{\lambda}_{n}(u)= \frac{1}{2\sqrt{n+\lambda}}\,\left[\sqrt{\frac{n(n+2\lambda-1)}{(n-1+\lambda)}}\, \mathcal{E}^{\lambda}_{n-1}(u) + \sqrt{\frac{(n+1)(n+2\lambda)}{(n+1+\lambda)}}\, \mathcal{E}^{\lambda}_{n+1}(u)\right]\,.
\end{equation}
The second one concerns the non-symmetric operator $(1-u^2)\ud/\ud u$.
\begin{equation}
\label{u2dEn}
(1-u^2)\frac{\ud}{\ud u}\, \mathcal{E}^{\lambda}_{n}(u)= \frac{\sqrt{n+\lambda}}{2}\,\left[\sqrt{\frac{n(n+2\lambda-1)}{(n-1+\lambda)}}\, \mathcal{E}^{\lambda}_{n-1}(u) - \sqrt{\frac{(n+1)(n+2\lambda)}{(n+1+\lambda)}}\, \mathcal{E}^{\lambda}_{n+1}(u)\right]\,.
\end{equation}

\subsection{Riesz bases and $\mathcal{D}$ quasi-bases}

The following notions related to biorthogonal sets have been mentioned along the paper. We give the main definitions here for readers' convenience.

\begin{defn}\label{def01}
A collection of vectors $\F_\varphi=\{\varphi_n, \,n\geq0\}$ in $\mathcal{H}$ is a Riesz basis for $\mathcal{H}$ if it is the image of an orthonormal basis for $\mathcal{H}$ under an invertible linear transformation. In other words, $\F_\varphi$ is a Riesz basis if there is an orthonormal basis $\{e_n\}$ for $\mathcal{H}$ and an invertible transformation $S$ such that $Se_n = \varphi_n$ for all $n$.
\end{defn}
In this case the set $\F_\psi=\{\psi_n=(S^{-1})^\dagger e_n,\, n\geq0\}$ is an Riesz basis as well, and it is biorthogonal to $\F_\varphi$: $\left<\varphi_n,\psi_m\right>=\delta_{n,m}$.

It is known that biorthogonal Riesz bases produce a resolution of the identity in $\Hil$. For physical reasons, \cite{baginbagbook}, it is sometimes convenient to consider the following weaker version of this resolution. This is what we get when dealing with $\D$-quasi bases, $\D$ being a dense subset of $\Hil$:

\begin{defn}\label{def02}
Two biorthogonal sets $\F_\eta=\{\eta_n\in\D,\,n\geq0\}$
and $\F_\Phi=\{\Phi_n\in\D,\,g\geq0\}$ are called {\em $\D$-quasi bases} if, for all $f, g\in \D$, the following holds: \be
\left<f,g\right>=\sum_{n\geq0}\left<f,\eta_n\right>\left<\Phi_n,g\right>=\sum_{n\geq0}\left<f,\Phi_n\right>\left<\eta_n,g\right>. \label{FB35}
\en
\end{defn}
Of course, when (\ref{FB35}) is satisfied, a weak resolution of the identity can be considered in $\D$.


\begin{thebibliography}{99}

\bibitem{carinena14} J. F. Cari\~nena, A. Ibort, G. Marmo, and G. Morandi,
\textit{Geometry from Dynamics, Classical and Quantum}, Springer, Dordrecht, Heildelberg, New York, London, 2015

\bibitem{ben1} C. M. Bender, S. Boettcher, {\em Real Spectra in Non-. Hermitian Hamiltonians Having PT-Symmetry}, Phys. Rev. Lett.,
{\bf 80}, 5243-5246, (1998)

\bibitem{curado} E. M. F. Curado, M. A. Rego-Monteiro, {\em Thermodynamic properties of a solid exhibiting the energy spectrum given by the logistic map},  Phys. Rev. E {\bf 61}, 6255-6260 (2000).

\bibitem{curado1} E. M. F. Curado, M. A. Rego-Monteiro, {\em Multi-parametric deformed Heisenberg algebras: a route to complexity},  J. Phys. A {\bf 34}, 3253-3264 (2001)

\bibitem{curado2} E. M. F. Curado, Y. Hassouni, M. A. Rego-Monteiro, Ligia M.C.S. Rodrigues, {\em Generalized Heisenberg algebra and algebraic method: The example of an infinite square-well potential},  Physics Letters A,  {\bf 372},  3350-3355 (2008)


\bibitem{baginbagbook} F. Bagarello, {\em Deformed canonical (anti-)commutation relations and non hermitian hamiltonians}, in {Non-selfadjoint operators in quantum physics: Mathematical aspects}, F. Bagarello, J. P. Gazeau, F. H. Szafraniec and M. Znojil Eds., John Wiley and Sons Eds.  (2015)

\bibitem{fbnlpbs} F. Bagarello, {\em Non linear pseudo-bosons}, J. Math. Phys.,  {\bf 52}, 063521, (2011)


\bibitem{rs} M. Reed and B. Simon, {\em Methods of Modern Mathematical Physics I: Functional analysis}, Academic Press, New York, (1980)

\bibitem{intop} Kuru S., Tegmen A., Vercin A., {\em Intertwined isospectral potentials in an arbitrary dimension},
J. Math. Phys, {\bf 42}, No. 8, 3344-3360, (2001); Kuru S.,
Demircioglu B., Onder M., Vercin A., {\em Two families of
superintegrable and isospectral potentials in two dimensions}, J.
Math. Phys, {\bf 43}, No. 5, 2133-2150, (2002); Samani K. A., Zarei
M., {\em Intertwined hamiltonians in two-dimensional curved spaces},
Ann. of Phys., {\bf 316}, 466-482, (2005); N. Aizawa, V. K. Dobrev, {\em Intertwining Operator Realization of Non-Relativistic Holography}, Nucl. Phys. B {\bf 828}, 581-593 (2010); B. Midya, B. Roy, R. Roychoudhury, {\em Position Dependent Mass Schroedinger Equation and Isospectral Potentials : Intertwining Operator approach}, J.  Math. Phys., {\bf 51}, 022109 (2010); A. L. Lisok, A. V. Shapovalov, A. Yu. Trifonov, {\em Symmetry and Intertwining Operators for the Nonlocal Gross-Pitaevskii Equation}, SIGMA {\bf 9}, 066, 21 pages (2013)


\bibitem{bagint1} F. Bagarello {\em Quons, coherent states and intertwining operators}, Phys. Lett.  A, {\bf 373}, 2637-2642 (2009)



\bibitem{bagint2} F. Bagarello, {\em Intertwining operators for non self-adjoint Hamiltonians and bicoherent states},  J. Math. Phys., {\bf 57}, 103501 (2016); doi: 10.1063/1.4964128

\bibitem{bagcomplex} F. Bagarello, {\em Non self-adjoint Hamiltonians with complex eigenvalues},  J. Phys. A, {\bf 49}, 215304 (2016)



\bibitem{antoine_etal_01} J.-P. Antoine,
J.-P. Gazeau, J. R. Klauder,  P. Monceau, and K. A. Penson,
Temporally stable coherent states for infinite well and P\"{o}schl–Teller potentials, J.  Math. Phys. \textbf{42} 2349 (2001).

\bibitem{magnus66}   Wilhelm Magnus, Fritz Oberhettinger, and Raj~Pal  Soni.
\newblock {\em Formulas and Theorems for
the Special Functions of Mathematical Physics}.
 \newblock Springer-Verlag,  Berlin, Heidelberg and New York, 1966.


























%\bibitem{jmll92} J.-M. L\'evy-Leblond, Elementary quantum models with position-dependent mass, Eur. J. Phys., \textbf{11} (1992), 215-218.
%
%\bibitem{jmll95} J.-M. L\'evy-Leblond, Position-dependent effective mass and Galilean invariance, Phys. Rev A, \textbf{52} (1995), 1845-1849.
%
%\bibitem{jmll74} J.-M L\'evy-Leblond, The pedagogical role and epistemological significance of group theory in quantum mechanics, Riv. Nuovo Cimento,  \textbf{4} (1974) 99-143 (1974).
%







\bibitem{bergasiyou10} H. Bergeron, J.-P. Gazeau, P. Siegl, and A. Youssef, Semi-classical behavior of Po ̈schl-Teller coherent states, EPL, \textbf{92} (2010) 60003

\bibitem{bersiyou12} H. Bergeron, P. Siegl, and A. Youssef, New SUSYQM coherent states for P\"{o}schl–Teller potentials: a detailed mathematical analysis,
J. Phys. A: Math. Theor. \textbf{45} (2012) 244028 (14pp)

\bibitem{RMECLRpra2017} M. A. Rego-Monteiro, E. M. F. Curado and Ligia M. C. S. Rodrigues, {\em Time evolution of linear and generalized Heisenberg algebra nonlinear P\"oschl-Teller coherent states}, Phys. Rev. A \textbf{96}, 0521221-0521229 (2017)
    
    
    
    
\bibitem{bagquons2}  F. Bagarello, {\em Deformed quons and bi-coherent states},   Proc. Roy. Soc. A, {\bf 473}, 20170049 (2017)

\bibitem{quons} A. Ballesteros, O. Civitarese, M. Reboiro, {\em Nonstandard $q$-deformed realizations of the harmonic oscillator},
Phys. Rev  C, {\bf 72}, 014305 (2005); A. Ballesteros, O. Civitarese, M. Reboiro, {\em Correspondence between the $q$-deformed harmonic oscillator and finite range potentials},
Phys. Rev  C, {\bf 68}, 044307 (2003)

\end{thebibliography}
\end{document}